\documentclass[review]{elsarticle}

\usepackage{lineno,hyperref}
\modulolinenumbers[5]
\usepackage{makeidx}
\makeindex

\usepackage{graphics} 
\usepackage{epsfig} 
\usepackage{times} 
\usepackage{amsmath} 
\usepackage{amssymb}  
\usepackage{adjustbox}
\usepackage{soul, color} 
\usepackage[english]{babel}
\usepackage{graphicx}
\usepackage{subcaption}
\usepackage{diagbox}
\usepackage{bbm}

\usepackage{upquote}
\usepackage{mathrsfs}
\usepackage{lmodern}
\usepackage{bm}
\usepackage{algorithm}
\usepackage{algorithmic}
\usepackage{booktabs,caption}
\usepackage[flushleft]{threeparttable}
\usepackage{multirow}

\usepackage{xpatch}
\usepackage[acronym,toc,nonumberlist,nogroupskip]{glossaries}
\renewcommand{\glossarysection}[2][]{}

\glssetwidest[0]{...................}

\makeglossaries

\usepackage{enumitem}

\usepackage{amsthm}
\newtheorem{thm}{Theorem}[section]
\theoremstyle{plain}

\newtheorem{lem}[thm]{Lemma}
\usepackage{thmtools}
\declaretheoremstyle[headfont=\normalfont]{normalhead}

\setlist{itemsep=0pt}

\newcommand{\splitatcommas}[1]{%
  \begingroup
  \begingroup\lccode`~=`, \lowercase{\endgroup
    \edef~{\mathchar\the\mathcode`, \penalty0 \noexpand\hspace{0pt plus 1em}}%
  }\mathcode`,="8000 #1%
  \endgroup
}

\usepackage{amsfonts}

\newcommand{\overbar}[1]{\mkern 1.5mu\overline{\mkern-1.5mu#1\mkern-1.5mu}\mkern 1.5mu}

\makeatletter
\def\ps@pprintTitle{%
 \let\@oddhead\@empty
 \let\@evenhead\@empty
 \def\@oddfoot{\centerline{\thepage}}%
 \let\@evenfoot\@oddfoot}

\begin{document}

\begin{frontmatter}

\title{Data-Driven Sensitivity Indices for Models With Dependent Inputs Using Polynomial Chaos Expansion{\color{black}s}}


\author[mymainaddress]{Zhanlin Liu}

\author[mymainaddress]{Youngjun Choe\corref{mycorrespondingauthor}}
\cortext[mycorrespondingauthor]{Corresponding author}
\ead{ychoe@uw.edu}

\address[mymainaddress]{Department of Industrial and Systems Engineering, University of Washington, Seattle, WA 98195, USA}

\begin{abstract}

Uncertainties exist in both physics-based and data-driven models. Variance-based sensitivity analysis characterizes how the variance of a model output is propagated from the model inputs. The Sobol index is one of the most widely used sensitivity indices for models with independent inputs. For models with dependent inputs, different approaches have been explored to obtain sensitivity indices in the literature. Typical approaches are based on procedures of transforming the dependent inputs into independent inputs. However, such transformation requires additional information about the inputs, such as the dependency structure or the conditional probability density functions. In this paper, data-driven sensitivity indices are proposed for models with dependent inputs. We first construct ordered partitions of linearly independent polynomials of the inputs. The modified Gram-Schmidt algorithm is then applied to the ordered partitions to generate orthogonal polynomials with respect to the empirical measure based on observed data of model inputs and outputs. Using the polynomial chaos expansion with the orthogonal polynomials, we obtain the proposed data-driven sensitivity indices. The sensitivity indices provide intuitive interpretations of how the dependent inputs affect the variance of the output without a priori knowledge on the dependence structure of the inputs. {\color{black}Four} numerical examples are used to validate the proposed approach.

\end{abstract}

 \end{frontmatter}

\section{Introduction}
Uncertainties exist in both physics-based and data-driven models. Uncertainty quantification (UQ) methods to characterize and reduce those uncertainties are increasingly popular in engineering studies. As an aspect of UQ, sensitivity analysis (SA) quantifies how output uncertainties are propagated from input uncertainties. Two general ways of conducting SA are local sensitivity analysis (LSA) and global sensitivity analysis (GSA). LSA analyzes how a small perturbation near an input space value could influence the output. On the contrary, GSA investigates how the input variability influences the output variability over the entire input space. In recent studies, variance-based sensitivity analysis, as a form of GSA, is utilized to understand system uncertainties in various applications such as material mechanics \cite{kala:2015sensitivity}, building energy \cite{sanchez:2014application}, structural mechanics \cite{xu2018cubature}, hydrogeology \cite{deman:2016using}, and manufacturing \cite{fesanghary2009design}.

Conducting variance-based sensitivity analysis for models with \textit{independent} inputs has been studied widely. Monte Carlo simulation and surrogate models are two general ways to obtain sensitivity indices for models with independent inputs. Surrogate models have been shown to be more computationally efficient compared with Monte Carlo simulation \cite{sudret:2012meta}. 
Polynomial chaos expansion (PCE) and Kriging (also known as Gaussian process regression) are the two surrogate models which have been used to compute sensitivity indices most commonly in the literature {\color{black}\cite{sudret:2012meta,le2017metamodel}}. Thanks to the orthogonal property of a PCE model, sensitivity indices for \textit{independent} inputs can be directly obtained using PCE coefficients {\color{black}\cite{sudret2006global,sudret2008global,le2017metamodel}}. PCE-based sensitivity indices appear in various fields including fluid dynamics \cite{Witteveen:2007}, structural reliability \cite{marelli2018active}, and vehicle dynamics \cite{Kewlani:2012}.

For models with \textit{dependent} inputs, a limited number of approaches are available in the literature to conduct variance-based sensitivity analyses. Generalized Sobol sensitivity indices have been proposed in Chastaing et al. \cite{chastaing2012generalized} based on the hierarchically orthogonal functional decomposition (HOFD). However, the unboundedness of the resulting sensitivity indices makes their interpretation for \textit{dependent} inputs not as straightforward as the Sobol indices for models with \textit{independent} inputs \cite{Navarro:2014}. A different framework is proposed in \cite{zhang2015new} to obtain sensitivity indices for models with correlated inputs. However, it requires the knowledge of model structure between the inputs and the outputs. An alternative way of obtaining sensitivity indices for models with \textit{dependent} inputs is to transform \textit{dependent} inputs into \textit{independent} inputs \cite{mara:2012, mara:2015non, tarantola2017variance}. Even though the transformation-based methods generate interpretable sensitivity indices, they require strong assumptions on the dependency or distributions of the inputs. 


The main contribution of this paper is the development of a data-driven method to obtain interpretable sensitivity indices for models with dependent inputs without invoking any assumptions on the inputs. 
We first propose the modified Gram-Schmidt based polynomial chaos expansion (mGS-PCE). The mGS-PCE increases the numerical robustness of constructing orthogonal polynomials for arbitrarily distributed inputs compared with the GS-PCE in \cite{Witteveen:2007}. 
Then, we propose a method to obtain data-driven sensitivity indices for models with dependent inputs by constructing ordered partitions of orthonormal polynomials of the inputs. This method estimates some of the sensitivity indices in \cite{mara:2012} and \cite{mara:2015non} without invoking the assumptions therein. Lastly, we propose conditional order-based sensitivity indices, which explain the model output variability in a hierarchical manner. 


The remainder of the paper is organized as follows: Section \ref{sec:2} reviews the background knowledge about Sobol indices and PCE models. Section \ref{sec:3} introduces the modified Gram-Schmidt algorithm and our data-driven method to obtain sensitivity indices for models with dependent inputs using PCE models. In Section \ref{sec:4}, {\color{black}four} numerical examples, where the inputs are dependent, are used to validate our proposed method. Section \ref{sec:5} provides a few concluding remarks and a discussion on future research directions.

\section{Technical background}
\label{sec:2}
This section briefly reviews sensitivity indices in the existing literature for models with independent inputs and those with dependent inputs. We first introduce the Hoeffding functional decomposition and the Sobol indices for independent inputs. We then review the \textit{full} sensitivity indices and the \textit{uncorrelated} sensitivity indices defined for models with \textit{dependent} inputs. Lastly, we introduce PCE models and explain how PCE coefficients can be used to calculate sensitivity indices.

\subsection{Hoeffding decomposition and sensitivity indices for independent inputs}

Suppose we have $n$ independent random inputs $\boldsymbol X = (X_{1}, X_{2}, \cdots, X_{n})$ with their density  $\mu(\boldsymbol X)$. For the output $Y = f(\boldsymbol X)$ that is square-integrable with respect to $\mu(\boldsymbol X)$, its Hoeffding decomposition 
is defined as follows \cite{sobol:1993,chastaing2012generalized}: 
\begin{equation}
\begin{aligned}
\label{eq:anova}  
f(\boldsymbol X) = \sum_{u \subseteq \{1,2,\ldots, n\}}f_{u}(\boldsymbol X_{u}), 
\end{aligned}
\end{equation}
where $f_{\emptyset} = f_{0}$ and $f_{0}$ is a constant and $\boldsymbol X_{u}$ = $\left( X_{j}\right)_{j \in u}$. Each summand $f_{u}(\boldsymbol X_{u}), u\neq \emptyset$, in Eq.~\eqref{eq:anova} satisfies 
\begin{equation*}
\begin{aligned}
\int f_{u}(\boldsymbol X_{u})\mu(X_{i}) \,\mathrm{d}X_{i} &= 0,\quad \forall i \in u. 
\end{aligned}
\end{equation*}
such that
\begin{equation*}
\begin{aligned}
f_{0} &= \int f(\boldsymbol X) \mu(\boldsymbol X) \,\mathrm{d}\boldsymbol X. \\ 
\end{aligned}
\end{equation*}
In addition, the summands in Eq.~\eqref{eq:anova} are orthogonal to each other as follows:
\begin{equation*}
\begin{aligned}
\int f_{u}(\boldsymbol X_u)f_{v}(\boldsymbol X_v) \mu(\boldsymbol X) \,\mathrm{d}\boldsymbol X &= 0,\quad \forall u, v \subseteq \{1,2,\ldots, n\}, v\neq u. 
\end{aligned}
\end{equation*}

 Based on the Hoeffding decomposition, the variance of $Y$ is decomposed as follows {\color{black}\cite{sudret2006global,sudret2008global}}:
\begin{equation*}
\begin{aligned}
Var(Y) &= \int f^{2}(\boldsymbol X)\mu(\boldsymbol X)\,\mathrm{d}\boldsymbol X - f_{0}^{2}\\
&= \sum_{\substack{ u \subseteq \{1,2,\ldots, n\} \\ u\neq \emptyset }}D_{u}(Y), \\
\end{aligned}
\end{equation*}
where 
\begin{equation}
\begin{aligned}
\label{eq:anova3}
D_{u}(Y) &= \int f^{2}_{u}(\boldsymbol X_{u})\mu(\boldsymbol X_{u}) \,\mathrm{d}\boldsymbol X_{u} \\
&=Var\!\left(E\!\left(Y\mid\boldsymbol X_{u}\right)\right) - \sum_{\substack{ v \subset u\\ v \neq u \\ v \neq \emptyset }}D_{v}(Y). 
\end{aligned}\nonumber
\end{equation}
For example, $D_{i}(Y) = Var\!\left(E\!\left(Y\mid X_{i}\right)\right)$ and $D_{ij}(Y) = Var\!\left(E\!\left(Y\mid X_{i}, X_{j}\right)\right)-D_{i}(Y)-D_{j}(Y)$.

Based on the variance decomposition, the Sobol index for set $u$ is defined as  
\begin{equation*}
\begin{aligned}
S_{u} = \frac{D_{u}(Y)}{Var(Y)},  
\end{aligned}
\end{equation*}
which measures the sensitivity of the output variance with respect to the inputs in $\boldsymbol X_{u}$.
For a particular input variable $X_{i}$, the \textit{first-order} Sobol index $S_{X_{i}}$ and \textit{total} Sobol index $ST_{X_{i}}$ are defined as follows: 
\begin{equation}
\begin{aligned}
\label{eq:sobtot}
S_{X_{i}} &= \frac{D_{i}(Y)}{Var(Y)},\\
ST_{X_{i}} &= \sum_{u \ni {i}}S_{u}.
\end{aligned} \nonumber
\end{equation}
$S_{X_{i}}$ represents the percentage of the output variance that is propagated from the input $X_{i}$. $ST_{X_{i}}$ represents the percentage of the output variance that is propagated from the input $X_{i}$ and its interactions with the other variables.  

\subsection{Sensitivity indices for dependent inputs}

This study focuses on sensitivity indices proposed in \cite{mara:2012, mara:2015non, tarantola2017variance} because they are bounded and do not require the knowledge of the model structure between the inputs and the output in contrast to those considered in \cite{kucherenko2012estimation, chastaing2012generalized,zhang2015new}, as discussed earlier.   

In \cite{mara:2012}, the Gram-Schmidt algorithm is employed to decorrelate the inputs when the dependences are characterized solely by the inputs' first-order conditional moments. Then the \textit{full} sensitivity indices and the \textit{uncorrelated} sensitivity indices (also called \textit{independent} sensitivity indices in \cite{mara:2015non}) are defined. On the other hand, in order to calculate these sensitivity indices when conditional probability density functions (cPDFs) of the inputs are known, the inverse Rosenblatt transformation or the inverse Nataf transformation is applied to transform the \textit{dependent} inputs into the \textit{independent} inputs \cite{mara:2015non, tarantola2017variance}. 


Suppose dependent inputs $(X_{1}, X_{2}, \ldots, X_{n} )$ are transformed into independent inputs $(\bar{X}_{1}, \bar{X}_{2}, \ldots, \bar{X}_{n} )${\color{black}, for example, under the assumptions of \cite{mara:2012} such that $\bar{X}_{1} = X_{1}$ and $\bar{X}_{i} = X_{i} - E\!\left(X_{i}\mid \bar{X}_{1}, \ldots, \bar{X}_{i-1}\right)$, $\forall i=2,\ldots,n$}.   {\color{black}Intuitively speaking, $\bar{X}_{1}$ keeps all information concerning $X_1$ including its dependent part with the other inputs. $\bar{X}_{i}$ contains all information concerning ${X}_{i}$ \emph{except} its dependent part with $\bar{X}_{1}, \ldots,\bar{X}_{i-1}$. Thus, $\bar{X}_{n}$ only contains information of ${X}_{n}$ \emph{excluding} its dependent part with all the other inputs. These constructed \emph{independent} inputs allow for calculating the first-order Sobol indices ($S_{\bar{X}_{i}}$) and total Sobol indices ($ST_{\bar{X}_{i}}$).} Then the sensitivity indices with respect to the {\color{black}\emph{dependent}} inputs are defined as follows \cite{mara:2012}:

\begin{itemize}
\item[] $\bar{S}_{X_{1}} = S_{\bar{X}_{1}} $ is the first-order full contribution of $X_{1}$ to the variance of the output. 
\item[] $\overbar{ST}_{X_{1}} = ST_{\bar{X}_{1}} $ is the total full contribution of $X_{1}$ to the variance of the output.
\item[]  $S^{u}_{X_{n}} = S_{\bar{X}_{n}} $ is the first-order uncorrelated contribution of $X_{n}$ to the variance of the output. 
\item[]  $ST^{u}_{X_{n}} = ST_{\bar{X}_{n}} $ is the total uncorrelated contribution of $X_{n}$ to the variance of the output.
\end{itemize}

By permuting the order of the inputs, different sensitivity indices can be further calculated. Suppose the initial input variables are ordered as $(X_{i},X_{i+1}, \ldots, X_{n}, \- X_{1}, \ldots, X_{i-1} )$, and the constructed independent inputs are $(\bar{X}_{i},\bar{X}_{i+1}, \ldots, \bar{X}_{n}, \bar{X}_{1},\ldots,\- \bar{X}_{i-1} )$. Then the \textit{full} sensitivity indices $(\bar{S}_{X_{i}} = S_{\bar{X}_{i}} \text{ and } \overbar{ST}_{X_{i}}= ST_{\bar{X}_{i}})$ and the \textit{uncorrelated} sensitivity indices $(S_{X_{i-1}}^{u} = S_{\bar{X}_{i-1}} \text{ and } ST_{X_{i}}^{u} = ST_{\bar{X}_{i-1}})$ are defined. $\bar{S}_{X_{i}}$ is called the first-order \textit{full} sensitivity index and $\overbar{ST}_{X_{i}}$ is called the total \textit{full} sensitivity index. $S_{X_{i}}^{u}$ is called the first-order \textit{uncorrelated} sensitivity index and $ST_{X_{i}}^{u}$ is called the total \textit{uncorrelated} sensitivity index.


\subsection{PCE and PCE-based sensitivity indices}

As a way of calculating sensitivity indices, PCE is known to be more computationally efficient than Monte Carlo simulations {\color{black}\cite{sudret2006global,sudret2008global}}. The original PCE, which is proposed in \cite{Wiener:1938}, provides Hermite polynomials for independent Gaussian random variables. Several types of PCE have been proposed under the assumption of \textit{independence} between model inputs, including the generalized PCE (gPCE) \cite{Xiu:2002}, the multi-element generalized PCE (ME-gPCE) \cite{Xiaoliang:2006}, the moment-based arbitrary PCE (aPCE) \cite{Oladyshkin:2012} and the Gram-Schmidt based PCE (GS-PCE) \cite{Witteveen:2007}.

The GS-PCE for models with \textit{independent} inputs is extended to models with multivariate \textit{dependent} inputs in Navarro et al. \cite{Navarro:2014}. It is regarded as the pioneering work in constructing an orthogonal polynomial basis for arbitrary \textit{dependent} inputs. 
Rahman \cite{rahman2018polynomial} theoretically validates the Gram-Schmidt orthogonalization process 
to construct an orthogonal polynomial basis for the PCE with dependent {\color{black} inputs}.  

\subsubsection{PCE model}
PCE uses a finite number of orthonormal polynomial terms of $n$ random inputs in $\boldsymbol X$ to approximate the output $Y$ as follows:
\begin{equation}
\label{eq:exp}
Y=f(\boldsymbol X) \approx \sum_{i=0}^{P}\theta_{i}\psi_{i}(\boldsymbol X),
\end{equation}
where $\theta_{i}$, $i$ = 0,1,2,\ldots,$P$, are called PCE coefficients and $\psi_{i}$, $i$ = 1,2,\ldots,$P$ are orthonormal polynomials. 
\begin{equation}
\label{eq:exp2}
P +1 = \binom{n+p}{n}
\end{equation}
%
is the number of polynomial terms, where
$p$ is the highest polynomial degree in the PCE model. As $p$ increases, the accuracy of approximating a complex output function improves. {\color{black} In this paper, we estimate the PCE coefficients by solving an overdetermined linear system of equations in the least-squares sense as proposed in \cite{berveiller2006stochastic}. 
}

{
Thanks to the properties of orthonormal polynomials, we can approximate the lower order moments of output $Y$ directly using the PCE coefficients in (\ref{eq:exp}) as follows: 

\begin{equation}
\begin{aligned}
\label{eq:mom}
E(Y) &\approx \theta_{0},\\
Var(Y) &\approx \sum_{i=1}^{P}\theta_{i}^2.
\end{aligned}
\end{equation}
}
The approximation errors converge to zero as $P$ increases \cite{Cameron:1947}.

\subsubsection{PCE-based sensitivity indices}

For \textit{independent} inputs, the multivariate orthonormal polynomials $\psi_{i}(\boldsymbol X)$ can be directly constructed as the products of univariate orthonormal polynomials as follows {\color{black}\cite{sudret2006global,sudret2008global,le2017metamodel}}: 
\begin{equation*}
\psi_{i} (\boldsymbol X) = \psi_{\boldsymbol{\alpha}_{i}}(\boldsymbol X) = \prod_{j=1}^{n}\psi_{\alpha_{ij}}(X_{j}), 
\end{equation*}
where $\boldsymbol{\alpha}_{i}=(\alpha_{i1}, \alpha_{i2}, \ldots, \alpha_{in})$ and $\psi_{\alpha_{ij}}(X_{j})$ represents the $\alpha_{ij}-$th order orthonormal polynomial in input $X_{j}$.

Define $\mathscr{A}_{u}$ as the set of multi-indices depending exactly on the subset of variables $\boldsymbol{X}_{u}, u \subseteq \{1,2,\ldots,n\}$ as follows:

\begin{equation*}
\label{eq:Au1}
\mathscr{A}_{u}= \left \{\boldsymbol{\alpha}_{i} \in \mathbb{N}^{n}: \alpha_{ij}\neq 0  \,{\color{black}\Leftrightarrow}\, j \in u, |\boldsymbol{\alpha}_{i}| \leq p \right \}, 
\end{equation*}
where 
\begin{equation*}
\label{eq:Au12}
|\boldsymbol \alpha_{i}| = \sum_{j=1}^{n}\alpha_{ij}. 
\end{equation*}

Suppose $\theta_{\boldsymbol{\alpha}_{i}}$ is the PCE coefficient with respect to the polynomial term corresponding to $\boldsymbol{\alpha}_{i}$. Then the first-order Sobol index for $X_j$ and the total Sobol index for $X_j$ can be estimated for $j=1,\ldots,n$ as follows {\color{black}\cite{sudret2006global,sudret2008global,le2017metamodel}}: 
\begin{equation}
\begin{aligned}
\label{eq:sobeqs}
S_{X_{j}} &\approx \frac{\sum_{\boldsymbol{\alpha}_{i} \in \mathscr{A}_{\{j\}}}\theta_{\boldsymbol \alpha_{i}}^{2}}{\sum_{i=1}^{P}\theta_{i}^{2}},\\ 
ST_{X_{j}} &\approx \sum_{\mathscr{A}_{u \ni j} } S_{\mathscr{A}_{u}},\\
\end{aligned} \nonumber
\end{equation}
where 
\begin{equation}
\begin{aligned}
\label{eq:sobeqs2}
S_{\mathscr{A}_{u}} &\approx \frac{\sum_{\boldsymbol{\alpha}_{i} \in \mathscr{A}_{u}}\theta_{\boldsymbol \alpha_{i}}^{2}}{\sum_{i=1}^{P}\theta_{i}^{2}}.
\end{aligned}\nonumber
\end{equation}
Using these PCE-based Sobol indices, we can also obtain the sensitivity indices for \textit{dependent} inputs, which were described earlier.

\section{Methodology}
\label{sec:3}
In the previous section, we discussed the current methods proposed in \cite{mara:2012} and \cite{mara:2015non} of obtaining the sensitivity indices for models with \textit{dependent} inputs under certain assumptions on the inputs. 
In this section, we propose a data-driven method to estimate the sensitivity indices for models with \textit{dependent} inputs using a PCE model based on the orthonormal polynomials constructed from the modified Gram-Schmidt algorithm. First, we show how to construct orthonormal polynomials using the modified Gram-Schmidt algorithm. Then we propose a data-driven method to estimate the \textit{first-order full} sensitivity indices and the \textit{total uncorrelated} sensitivity indices for models with \textit{dependent} inputs. Then we propose an alternative \textit{total full} sensitivity index and an alternative \textit{first-order uncorrelated} sensitivity index, which can be also calculated using the proposed method. These alternative indices have different interpretations than those in \cite{mara:2012} and \cite{mara:2015non} because our decorrelation process does not eliminate dependences in inputs. In addition, we propose conditional order-based sensitivity indices and illustrate how they can be used to reduce the PCE model complexity by excluding higher order interaction terms.

\subsection{Modified Gram-Schmidt algorithm}

In \cite{Navarro:2014}, orthonormal polynomials are constructed using the Gram-Schmidt algorithm for general multivariate 
correlated variables. Even though the Gram-Schmidt algorithm behaves the same as the modified Gram-Schmidt algorithm mathematically, the modified Gram-Schmidt algorithm is less sensitive to numeric rounding errors and performs more stably than the Gram-Schmidt algorithm \cite{bjorck:1994}. Therefore, we propose to use the modified Gram-Schmidt algorithm to construct orthonormal polynomial basis $\{\psi_{i}(\boldsymbol X)\}_{i=1}^{P}$ based on the initial $P$ linearly independent polynomials $\left(e_{i}\right)_{i \in\{1,2,\ldots, P\}}$ as follows \cite{bjorck1992loss}:
\begin{algorithm}[H]
\caption{Modified Gram-Schmidt Algorithm}
\begin{algorithmic}[1] 
     \FOR{$i = 1, 2, \ldots, P $}
	\STATE $\phi_{i}(\boldsymbol X) \leftarrow e_{i}(\boldsymbol 	X)$
	\FOR{$k = 1, 2, \ldots, i-1 $}
\STATE{$\phi_{i}(\boldsymbol X) \leftarrow  \phi_{i}(\boldsymbol X) - \langle \phi_{i}(\boldsymbol X), \psi_{k}(\boldsymbol X) \rangle \psi_{k}(\boldsymbol X) $}
      	\ENDFOR
 	\STATE $\psi_{i}(\boldsymbol X)  \leftarrow\frac{\phi_{i}(\boldsymbol X)}{||\phi_{i}(\boldsymbol X)||_{2} }$
	\ENDFOR

\end{algorithmic}
\end{algorithm}
\noindent The inner product in the algorithm is defined with respect to the empirical measure in this paper. {\color{black}The inner product is numerically evaluated using the observations of $\boldsymbol X$ in a given dataset. Note that the proposed data-driven method assumes neither any distributional knowledge of $\boldsymbol X$ nor the ability to easily sample from its distribution. Thus, we do not use a Monte Carlo approach to evaluate the inner product although it may be an option for the problems that permit the sampling.}

The difference between the standard Gram-Schmidt algorithm and the modified Gram-Schmidt algorithm is at the line 4 in Algorithm 1, where the standard Gram-Schmidt algorithm performs 
\[
\phi_{i}(\boldsymbol X) \leftarrow  \phi_{i}(\boldsymbol X) - \left\langle e_{i}(\boldsymbol X), \psi_{k}(\boldsymbol X) \right\rangle \psi_{k}(\boldsymbol X).
\]

Note that different orthonormal polynomials are constructed from different permutations of the initial polynomials. In the following section, we discuss how to permute the order of the initial polynomials in order to obtain data-driven sensitivity indices for models with \textit{dependent} inputs.

\subsection{Sensitivity indices}
As we discussed in the previous section, PCE models can be constructed for models with \textit{dependent} inputs based on the modified Gram-Schmidt algorithm. In this section, we first propose how to use PCE models to estimate the \textit{full} sensitivity indices and the \textit{uncorrelated} sensitivity indices based on data. Then we define the conditional order-based sensitivity indices and present how they can be used to exclude higher order interaction terms in a PCE model. For easy reference, we include in Appendix {\color{black}A.1} a list of sensitivity index symbols used in this paper.

\subsubsection{Full sensitivity indices}\label{sec:full_SI}

Constructing orthonormal polynomials using the modified Gram-Schmidt algorithm requires a linearly independent set of polynomials. A PCE model with $n$ inputs 
and the highest polynomial order $p$ is composed of $P+1$ terms of polynomials as we defined in Eqs.~\eqref{eq:exp} and \eqref{eq:exp2}. Assume polynomials in the set 
\begin{equation}
\label{eq:poly}
S = \left\{ \prod_{l=1}^{n}X_{l}^{j_{l}} : 
j_{l} \in \{0,1,\ldots p\}, \sum_{l=1}^{n}j_{l} \leq p\right\} 
\end{equation}
are linearly independent. 

Orthonormal polynomials can be constructed using the modified Gram-Schmidt algorithm with respect to a specific order of the polynomials. 
Suppose we order the polynomials in $S$ as $(St_{0}, St_{11}, St_{1}{\color{black}\setminus }St_{11}, St_{2}, St_{3},\- \ldots, St_{n})$, where $St_{0} = \{1\}$, $St_{11}$ and $St_{i}$ are defined as follows:
\begin{equation}
\begin{aligned}
\label{eq:order1}
&St_{11} = \left\{X_{1}^{j_{1}}: j_{1} \in \{1,\ldots p\} \right\}, \\
&St_{i} = \left(S{\color{black} \setminus} \bigcup_{j=0}^{i-1} St_{j}\right) \bigcap 
\left \{\prod_{l=1}^{n}X_{l}^{j_{l}}: j_{i} \in \{1,2,\ldots, p\} , j_{l \neq i} \in \{0, 1,\ldots, p\}, \sum_{l=1}^{n}j_{l} \leq p \right\}. 
\end{aligned}
\end{equation}
%
$\left(St_{0}, St_{11}, St_{1} {\color{black} \setminus} St_{11}, St_{2}, St_{3},\- \ldots, St_{n}\right)$ is an ordered partition of the set $S$. In the partition, $St_{0}$, $St_{11}$, $St_{1}{\color{black} \setminus}St_{11}$, and $St_{i}, i = 2, 3, \ldots, n$ are ordered in sequence but the polynomials in each set can be in any arbitrary order. Note that $St_{1}$ contains all the polynomial functions of $X_{1}$ and the interaction terms between $X_{1}$ and the rest of the inputs. $St_{2}$ contains all the polynomial functions of $X_{2}$ and the interaction terms between $X_{2}$ and the rest of the inputs \emph{except} $X_{1}$. $St_{3}$ contains all the polynomial functions of $X_{3}$ and the interaction terms between $X_{3}$ and the rest of the inputs \emph{except} $X_{1}$ and $X_{2}$.

For example, $St_{11}$ and $St_{i}, i = 1,2, 3$ for constructing a PCE model with inputs $\left\{X_{1}, X_{2}, X_{3}\right\}$ and the highest polynomial order $p=3$ are defined as follows: 
\begin{equation}
\label{eq:lstexample}
\begin{aligned}
&St_{11} = \left\{X_{1}, X_{1}^{2}, X_{1}^{3} \right\}, \\
&St_{1} = \left\{X_{1}, X_{1}^{2}, X_{1}^{3}, X_{1}X_{2}, X_{1}^{2}X_{2}, X_{1}X_{2}^{2}, X_{1}X_{3}, X_{1}^{2}X_{3}, X_{1}X_{3}^{2}, X_{1}X_{2}X_{3} \right\}, \\
&St_{2} = \left\{X_{2}, X_{2}^{2}, X_{2}^{3}, X_{2}X_{3}, X_{2}^{2}X_{3}, X_{2}X_{3}^{2} \right\}, \\
&St_{3} =\left\{X_{3}, X_{3}^{2}, X_{3}^{3}\right\}. 
\end{aligned} \nonumber
\end{equation}
%

After constructing the orthonormal polynomials using the ordered partition $(St_{0}, St_{11}, St_{1} {\color{black}\setminus} St_{11}, St_{2}, St_{3},\- \ldots, St_{n})$, the first-order full sensitivity index $\bar{S}_{X_{1}}$ for $X_{1}$ can be estimated as follows:
\begin{equation}
\begin{aligned}
\label{eq:fulleq}
\bar{S}_{X_{1}} \approx \frac{\sum_{j \in St_{11}}\theta_{j}^{2}}{Var(Y)},\\
\end{aligned}
\end{equation}
where $\theta_{j}$'s are the PCE coefficients corresponding to the orthonormal polynomials in the set $St_{11}$. 

In addition, we propose an alternative total full sensitivity index $$\overbar{ST}_{X_{1}} = \frac{\sum_{u \ni X_{1}}D_{u}(Y)}{Var(Y)}$$
and estimate it using
\begin{equation}
\begin{aligned}
\label{eq:fulltot}
\overbar{ST}_{X_{1}} \approx \frac{\sum_{j \in St_{1}}\theta_{j}^{2}}{Var(Y)}.
\end{aligned}
\end{equation}
This total full sensitivity index is different from the one defined in \cite{mara:2012}, which is obtained after transforming the dependent inputs into the independent inputs. Instead, the total full sensitivity index in Eq.~\eqref{eq:fulltot} has dependent effects of $X_{1}$ with other inputs. By permuting the order of the input $(X_{1}, X_{2}, X_{3}, \ldots, X_{n})$ as $(X_{i}, X_{i+1}, \ldots, X_{n},X_{1},\- \ldots, X_{i-1})$, any $\bar{S}_{X_{i}}$ and $\overbar{ST}_{X_{i}}$ can be estimated.


We also define the conditional total sensitivity indices for models with \textit{dependent} inputs as follows: 
\begin{itemize}
\setlength\itemsep{-2em}
\item[] $\overbar{ST}_{X_{2}|X_{1}} = \frac{\sum_{u \ni \{X_{1} ,X_{2}\} }D_{u}(Y) - \sum_{u \ni X_{1}}D_{u}(Y) }{Var(Y)}$ is the total contribution of input $X_{2}$ to the variance of output $Y$ after taking account of the total full contribution of $X_{1}$. \\
\item[] $\overbar{ST}_{X_{3}|X_{1},X_{2}} = \frac{\sum_{u \ni \{X_{1} ,X_{2}, X_{3}\} }D_{u}(Y) - \sum_{u \ni \{X_{1}, X_{2}\}}D_{u}(Y) }{Var(Y)}$ is the total contribution  of input $X_{3}$ to the variance of output $ Y$ after taking account of the total full contributions of $X_{1}$ and $X_{2}$. \\
\item[] $\vdots$ \\
\item[] $\overbar{ST}_{X_{n}|X_{1},X_{2},\ldots,X_{n-1}} = \frac{\sum_{u \ni \{X_{1} ,X_{2}, \ldots, X_{n} \} }D_{u}(Y) - \sum_{u \ni \{X_{1}, X_{2}, \ldots, X_{n-1}\}}D_{u}(Y) }{Var(Y)}$ is the total contribution of input $X_{n}$ to the variance of output $ Y$ after taking account of the total full contributions of $X_{1}$, $X_{2}$, \ldots, $X_{n-1}$. 
\end{itemize}
We estimate the conditional total sensitivity indices using
$$\overbar{ST}_{X_{i}|X_{1},X_{2},\ldots,X_{i-1}} \approx \frac{\sum_{j \in St_{i}}\theta_{j}^{2}}{Var(Y)}$$
for $i=2,3,\ldots,n$

When inputs can be grouped such that inputs from different groups have neither dependence nor interaction across groups, the total Sobol index of a group can be estimated based on the conditional total sensitivity indices. For example, if the first $d$ inputs are independent of and have no interactions with the rest of the inputs, $\sum_{i=1}^{d} \overbar{ST}_{X_{i}}$ estimates the total Sobol index of the first $d$ inputs. Eq.~\eqref{eq:exeq1} for Example 3 in Section~\ref{sec:4} illustrates how this sensitivity index can be used in practice.

\subsubsection{Uncorrelated sensitivity indices}
In order to estimate the total \textit{uncorrelated} sensitivity index of $X_{1}$, we consider the ordered partition of $S$ as $(St_{0}, St_{-1}, St_{11}, St_{1}{\color{black}\setminus} St_{11})$, where $St_{0} = \{1\}$ and $St_{-1}= \bigcup_{j=2}^{n}St_{j}$. $St_{11}$ and $St_{i}, i=1,2,\ldots,n $ are defined in Eq.~\eqref{eq:order1}. Then the orthonormal polynomials can be constructed with respect to this ordered partition.   
As we obtain the PCE {\color{black}coefficients} corresponding to the orthonormal polynomials in the set $St_{1}$, the total \textit{uncorrelated} sensitivity index ($ST_{X_{1}}^{u}$) can be estimated using Eq.~\eqref{eq:fulltot}. 

In addition, we propose an alternative first-order \textit{uncorrelated} sensitivity index ($S_{X_{1}}^{u}$) and estimate it using Eq.~\eqref{eq:fulleq}. The proposed first-order \textit{uncorrelated} sensitivity index is different from the one defined in \cite{mara:2012} because the latter is estimated after decorrelating $X_{1}$ with all the other inputs. Note that the proposed first-order \textit{uncorrelated} sensitivity index is estimated by decorrelating the polynomials of the inputs. 
By permuting the order of the inputs  $(X_{1}, X_{2}, X_{3}, \ldots, X_{n})$ as $(X_{i}, X_{i+1}, \ldots, \- X_{n},X_{1}, \ldots, X_{i-1})$, any $S_{X_{i}}^{u}$ and $ST_{X_{i}}^{u}$ can be {\color{black}estimated}.

Note that the first-order full ({\color{black}resp.} uncorrelated) sensitivity indices are always smaller or equal to the total full ({\color{black}resp.} uncorrelated) sensitivity indices, but the first-order full ({\color{black}resp.} uncorrelated) sensitivity indices are not necessarily larger or smaller than the total uncorrelated ({\color{black}resp.} full) sensitivity indices.

\subsubsection{Conditional order-based sensitivity indices}
In order to reduce the complexity of a PCE model and select appropriate interaction terms in the PCE model, we propose the conditional order-based sensitivity indices. 

Suppose we order the polynomials in the polynomial set $S$ in Eq.~\eqref{eq:poly} as $\left(Sc_{0}, Sc_{11}, Sc_{12}, \ldots, Sc_{1p} , Sc_{22}, Sc_{23}, \ldots, Sc_{2p} , \ldots, Sc_{kk}, Sc_{kk+1}, \right. \\ \left. \ldots, Sc_{kp} \right)$, where $k = \min(n,p)$, $Sc_{0} =\{1\}$, and $Sc_{ij}, i= 1, 2, \ldots, k$; $j = i, i+1, \ldots, p$, are defined as follows: 
\begin{equation*}
\label{eq:order3}
Sc_{ij} = \left\{\prod_{l=1}^{n}X_{l}^{j_{l}}:  
j_{l} \in \{0, 1,\ldots p\}, \sum_{l=1}^{n}\mathbbm{1}_{j_{l} \neq 0} = i, \sum_{l=1}^{n}j_{l} =j \right\}, 
\end{equation*}
where 
\begin{equation*}
\label{eq:case1}
\mathbbm{1}_{j_{l} \neq 0} = \begin{cases}
1 & j_{l} \neq 0\\
0 & j_{l} = 0 
\end{cases}.
\end{equation*}

Define $Sc_{i}=\cup_{j=i}^{p}Sc_{ij}, i \leq \min(n, p)$, then $Sc_{1}$ contains all the polynomial functions of $X_{i}, i=1,2,\ldots, n $. $Sc_{2}$ contains all the two-way interaction terms. $Sc_{3}$ contains all the three-way interaction terms. Note that $\left(Sc_{0}, Sc_{11}, Sc_{12}, \ldots, Sc_{1p}, \right.\\ \left. Sc_{22}, Sc_{23}, \ldots, Sc_{2p}, \ldots, Sc_{kk}, Sc_{kk+1}, \ldots, Sc_{kp} \right)$, where $k = \min(n,p)$, is an ordered partition of the set $S$. In the partition, $Sc_{0}, Sc_{11}, Sc_{12}, \ldots, \-Sc_{1p}, Sc_{22},\- Sc_{23}, \ldots, Sc_{2p}, \ldots,\- Sc_{kk}, Sc_{kk+1}, \ldots, Sc_{kp}$, are ordered in sequence but the polynomials in each set $Sc_{ij}$ can be in any arbitrary order. 

For example, 
$Sc_{11}, Sc_{12}, Sc_{13}, Sc_{22}, Sc_{23}, \text{ and } Sc_{33}$ for constructing a PCE model with inputs $\left\{X_{1}, X_{2}, X_{3}\right\}$ and the highest polynomial order $p=3$ are defined as follows: 
\begin{equation}
\label{eq:lstexample2}
\begin{aligned}
&Sc_{11} = \left\{X_{1}, X_{2}, X_{3} \right\}, \\
&Sc_{12} = \left\{X_{1}^{2}, X_{2}^{2}, X_{3}^{2} \right\}, \\
&Sc_{13} = \left\{X_{1}^{3}, X_{2}^{3}, X_{3}^{3} \right\}, \\
&Sc_{22} = \left\{X_{1}X_{2}, X_{2}X_{3}, X_{1}X_{3} \right\}, \\
&Sc_{23} = \left\{X_{1}^{2}X_{2}, X_{1}^{2}X_{3}, X_{1}X_{2}^{2}, X_{1}X_{3}^{2}, X_{2}^{2}X_{3}, X_{2}X_{3}^{2},\right\}, \\
&Sc_{33} = \left\{X_{1}X_{2}X_{3} \right\}.  \\
\end{aligned}\nonumber
\end{equation}

We define the conditional order-based sensitivity indices as follows: 

\begin{itemize}
\setlength\itemsep{-2em}
\item[] $\tilde{S}_{1} = \frac{\sum_{i=1}^{n}D_{i}(Y)}{Var(Y)}$ is the first order sensitivity index of the output $Y$ with respect to the inputs $\boldsymbol X$.  \\
\item[] $\tilde{S}_{2|1}= \frac{\sum_{i<j}D_{ij}(Y)}{Var(Y)}$ is the second order sensitivity index of the output $Y$ with respect to the inputs $\boldsymbol X$ after taking account of the first order sensitivity index. \\
\item[] $\tilde{S}_{3|1,2}= \frac{\sum_{i<j<k}D_{ijk}(Y)}{Var(Y)}$ is the third order sensitivity index of the output $Y$ with respect to the inputs $\boldsymbol X$ after taking account of the first order sensitivity index and the second order sensitivity index. \\
\item[] $\vdots$ \\
\item[] $\tilde{S}_{k|1,2,\ldots,k-1}= \frac{D_{1,2,3,\ldots, k}(Y)}{Var(Y)}$ is the $k^{th}$ order sensitivity index of the output $Y$ with respect to the inputs $\boldsymbol X$ after taking account of the first order sensitivity index through the $(k-1)^{th}$ order sensitivity index. \\
\end{itemize}
As we can obtain the PCE coefficients corresponding to the orthonormal polynomials constructed from $\left(Sc_{0}, Sc_{11}, Sc_{12}, \ldots, Sc_{1p} , Sc_{22},\- Sc_{23}, \ldots, Sc_{2p} , \ldots,\right. \\ \left. Sc_{kk}, Sc_{kk+1}, \ldots, Sc_{kp} \right)$, where $k = \min(n,p)$, using the modified Gram-Schmidt algorithm, the above sensitivity indices can be estimated as follows: 

\begin{equation}
\begin{aligned}
\label{eq:Sob_cal3}
\tilde{S}_{1} &\approx \frac{\sum_{j \in Sc_{1}} \theta_{j}^{2}}{Var(\boldsymbol Y )}, \\
\tilde{S}_{i|1,\ldots,i-1} &\approx \frac{\sum_{j \in Sc_{i}} \theta_{j}^{2}}{Var(\boldsymbol Y )}, \,\,\, i=2,\ldots,\min(n,p),\\
\end{aligned}\nonumber
\end{equation}
%
where $\theta_{j}$'s are the PCE coefficients corresponding to the orthonormal polynomials in the set $Sc_{i}$ for $i \leq \min(n,p)$. Note that for a full PCE model, $Sc_{i}$ contains $\binom{n}{i} \binom{p}{i} $ polynomial terms.

{\color{black}The conditional order-based sensitivity indices serve the purpose of identifying up to which \emph{order of interaction} of inputs significantly influences the output variance. Specifically,}
if the cumulative sum of conditional order-based sensitivity indices, $\sum_{i=1}^{d}\tilde{S}_{i}$, is close to one for a certain polynomial order, $d\leq \min(n,p)$, it indicates that the interaction terms of orders higher than $d$ can be excluded in the PCE model. Using a simple procedure of inspecting the cumulative sum for different $d$'s, we can identify and remove unnecessary high-order interaction terms from the PCE model. In constrast to the existing methods of constructing a sparse PCE model \cite{blatman2008sparse,blatman2011adaptive, jakeman2015enhancing, peng2016polynomial}, this procedure keeps the effect hierarchy principle \cite{wu2011experiments} while improving the parsimony of the PCE model. 
Example 3 in Section~\ref{sec:4} illustrates how this procedure can be used to determine the highest polynomial order. 

{\color{black}While the conditional order-based sensitivity indices are useful for effective PCE modeling and, in turn, PCE-based sensitivity analyses, the new sensitivity indices do not directly serve the traditional purposes of sensitivity indices. In contrast, for example, the first-order full sensitivity indices and the total uncorrelated sensitivity indices directly help determine influential and non-influential inputs, respectively, in terms of their contributions to the output variance (also known as factor prioritization and factor fixing, respectively, in \cite{saltelli2008global}; see also \cite{mara:2012}).}

\section{Numerical examples}
\label{sec:4}
To validate the proposed data-driven sensitivity indices, this section presents {\color{black}four} examples where inputs are dependent. {\color{black}We present our experiment results based on 500 replications with 95\% confidence intervals wherever applicable. The confidence intervals are computed using $10,000$ bootstrap samples of the 500 replications to improve upon the accuracy of the empirical confidence interval computed from the 500 replications \cite{diciccio1996bootstrap}.}  

\subsection{Example 1}
We use a benchmark example in \cite{mara:2012} as our first validation case. In this case, inputs follow a three-dimensional multivariate normal distribution as follows:
\begin{equation}
\begin{aligned}
\begin{pmatrix}X_{1}\\
X_{2}\\
X_{3}
\end{pmatrix} &\sim \mathcal{N}
\begin{bmatrix}
\begin{pmatrix}
0\\
0\\
0
\end{pmatrix}\!\!,&
\begin{pmatrix}
1 & \rho_{12} & \rho_{13}\\
\rho_{12} & 1 & \rho_{23}\\
\rho_{13} & \rho_{23} & 1
\end{pmatrix}
\end{bmatrix}.\\[2\jot]
\end{aligned}\nonumber
\end{equation}
The output $Y$ is simply modeled using a linear model $Y=X_{1}+X_{2}+X_{3}$. 

\begin{table}[ht]
\centering
\caption{Sample mean of sensitivity indices {\color{black}and $95\%$ confidence intervals.}} 
\begin{adjustbox}{width=0.99\textwidth}
\label{tab:example_1}
\begin{tabular}{|c|c|cc|cc|}
\hline
\multirow{2}{*}{$(\rho_{12}, \rho_{13}, \rho_{23})$} & \multirow{2}{*}{Input} & \multicolumn{2}{c}{$\bar{S}_{X_{i}}$} & \multicolumn{2}{c|}{$ST^{u}_{X_{i}}$} \\ 
 \cline{3-6}
 &  & $\text{Analytical method}^\dagger$ & $\text{Proposed method}^\ddagger$ & $\text{Analytical method}^\dagger$ & $\text{Proposed method}^\ddagger$ \\
  \hline
(0.5,0.8,0) & $X_{1}$ & {\color{black}$0.945$} & {\color{black}$0.945$ $(0.945, 0.945)$} & {\color{black}$0.020$} & {\color{black}$0.020$ $(0.020, 0.020)$}   \\
 & $X_{2}$ & {\color{black}$0.402$} & {\color{black}$0.401$ $(0.400, 0.402)$}  & {\color{black}$0.055$} & {\color{black}$0.055$ $(0.054, 0.055)$} \\
 & $X_{3}$ & {\color{black}$0.579$} & {\color{black}$0.579$ $(0.578, 0.579)$} & {\color{black}$0.026$} & {\color{black}$0.026$ $(0.026,0.026)$}\\\hline
(-0.5,0.2,-0.7) & $X_{1}$ & {\color{black}$0.490$} & {\color{black}$0.491$ $(0.490,0.492)$} & {\color{black}$0.706$}& {\color{black}$0.707$ $(0.706,0.707)$} \\
 & $X_{2}$ & {\color{black}$0.040$} & {\color{black}$0.041$ $(0.040,0.041)$} & {\color{black}$0.375$} & {\color{black}$0.374$ $(0.373,0.374)$} \\
 & $X_{3}$ & {\color{black}$0.250$} & {\color{black}$0.250$ $(0.249,0.251)$} & {\color{black}$0.480$} & {\color{black}$0.480$ $(0.479,0.481)$} \\\hline
(-0.49,-0.49,-0.49) & $X_{1}$ &  {\color{black}$0.007$} & {\color{black}$0.007$ $(0.007, 0.007)$} &  {\color{black}$0.974$} & {\color{black}$0.974$ $(0.974, 0.974)$} \\
 & $X_{2}$ & {\color{black}$0.007$ } & {\color{black}$0.007$ $(0.006, 0.007)$} &  {\color{black}$0.974$} & {\color{black}$0.974$ $(0.974, 0.974)$} \\
 & $X_{3}$ &{\color{black}$0.007$} & {\color{black}$0.007$ $(0.007, 0.007)$} & {\color{black}$0.974$} & {\color{black}$0.974$ $(0.973, 0.974)$} \\ \hline
\end{tabular}
\end{adjustbox}
\begin{tablenotes}
      \small
      \item \textit{Note:} {$\dagger$The value presented in \cite{mara:2012} differs by up to 0.01 due to rounding. $\ddagger${\color{black}The  proposed  method  does  not  require  any  distributional  assumption.} The sample mean{\color{black}s of the sensitivity indices and $95\%$ bootstrap confidence intervals (using 10,000 bootstrap samples) are calculated} based on $500$ simulation replications where each replication uses $500$ random observations.  } 
    \end{tablenotes}
\end{table}


Table \ref{tab:example_1} shows the first-order full sensitivity index $\bar{S}_{X_{i}}$ and total uncorrelated sensitivity index $ST^{u}_{X_{i}}$ for each input based on the analytical method \citep{mara:2012}. These true indices are compared with the estimated  indices from the proposed method.  This example  validates that the proposed method can estimate the first-order full sensitivity index and total uncorrelated sensitivity index based only on data without the knowledge of the distribution of dependent inputs and the model structure. Note that in this example,  the \emph{total} full sensitivity index is the same as the \emph{first-order} full sensitivity index (i.e., $\overbar{ST}_{X_{i}} = \bar{S}_{X_{i}}$) and that the \emph{first-order} uncorrelated sensitivity index equals the \emph{total} uncorrelated sensitivity index (i.e., $S^{u}_{X_{i}}=ST^{u}_{X_{i}}$) because there is no interaction effect. Thus, $\overbar{ST}_{X_{i}}$ and $S^{u}_{X_{i}}$ are not presented.

\subsection{Example 2}
In this example from \cite{tarantola2017variance}, the output $Y$ is a non-linear function of four dependent inputs: $Y= X_{1}X_{2} + X_{3}X_{4}$. Here, $(X_{1},X_{2}) \in \big[0,1\big]^{2}$ is uniformly distributed within the triangle $X_{1}+X_{2}\leq 1$ and $(X_{3},X_{4}) \in \big[0,1\big]^{2}$ is uniformly distributed within the triangle $X_{1}+X_{2}\geq 1$. Due to the symmetry of the model, the sensitivity indices of $Y$ with respect to $X_{1}$ and $X_{3}$ are equal to those with respect to $X_{2}$ and $X_{4}$, respectively. 

\begin{table}[ht]
\centering
\caption{Sample mean of sensitivity indices {\color{black}and $95\%$ confidence intervals.}} 
\begin{adjustbox}{width=0.99\textwidth}
\label{tab:example_2}
\begin{tabular}{|c|cccccc|}
\hline
Method & $\bar{S}_{X_{1}}$ & $ST^{u}_{X_{2}}$ & $ST_{\{X_{1},X_{2}\}}$ & $\bar{S}_{X_{3}}$ & $ST^{u}_{X_{4}}$ & $ST_{\{X_{3},X_{4}\}}$ \\ \hline
$\text{Analytical method}^{\dagger}$ & $0.033$ & $0.067$ & $0.100$ & $0.233$ & $0.666$ & $0.900$ \\ \hline
\multirow{2}{*}{$\text{Benchmark method}^{\ddagger}$} & $0.032$ & $0.071$ & $0.103$ & $0.226$ & $0.669$ & $0.895$ \\
 & (0.028, 0.037) & (0.066,0.077) & (0.095,0.114) & (0.209,0.248) & (0.639,0.705) & (0.848, 953)\\
 \hline
\multirow{2}{*}{$\text{Proposed method}^*$} & $0.035$ & $0.066$ & $0.101$ & $0.233$ & $0.663$ & $0.896$ \\
 & (0.034, 0.036) & (0.066,0.067) & (0.100,0.103) & (0.231,0.235) & (0.661,0.666) & (0.892, 901)\\
 \hline
\end{tabular}
\end{adjustbox}
\begin{tablenotes}
      \small
      \item \textit{Note:} {$\dagger$The value is provided in \cite{tarantola2017variance}. $\ddagger$The value is estimated using the method proposed in \cite{tarantola2017variance} and the confidence interval is calculated based on 16,380 random observations under the assumption that the joint probability distribution of the inputs is \emph{known}. *The proposed method does not require any distributional assumption. The sample means of sensitivity indices and $95\%$ bootstrap confidence intervals (using $10,000$ bootstrap samples) are calculated based on 500 replications. In each replication, sensitivity indices are calculated using 500 random observations.   } 
    \end{tablenotes}
\end{table}



As shown in Table \ref{tab:example_2}, the proposed method yields the estimates that are close to the analytical values of $\bar{S}_{X_{i}}$ and $ST^{u}_{X_{i}}$. In contrast to the benchmark method in \cite{tarantola2017variance} that requires the knowledge of joint probability distribution of the inputs, the proposed method is purely data-driven. $ST_{\{X_{1},X_{2}\}}$ (or, $ST_{\{X_{3},X_{4}\}}$) is estimated using $\sum_{i=1}^{2}\overbar{ST}_{X_{i}}$ by permuting the inputs as $(X_{1}, X_{2}, X_{3}, X_{4})$ (or, $(X_{3}, X_{4}, X_{1}, X_{2})$).


Figure \ref{fig:exampel2} shows intricate working of the model by revealing how each input influences the output variance. The \textit{total} full (uncorrelated) sensitivity index $\overbar{ST}_{X_{i}}$ $(ST^{u}_{X_{i}})$ can be decomposed into the \textit{first-order} full (uncorrelated) sensitivity index $\bar{S}_{X_{i}}$ $(S^{u}_{X_{i}})$ and the \textit{rest of the total} effect, $\overbar{ST}_{X_{i}} - \bar{S}_{X_{i}}$ $(ST^{u}_{X_{i}} - S^{u}_{X_{i}} )$, which accounts for all the interactions of $X_{i}$. The gap between the two lines on the left (right) graph in Figure~\ref{fig:exampel2} shows the magnitude of $\overbar{ST}_{X_{i}} - \bar{S}_{X_{i}}$ $(ST^{u}_{X_{i}} - S^{u}_{X_{i}} )$, indicating how much of the total effect of $X_{i}$ is attributed to the interaction effects compared to the first-order effect when we consider the full (uncorrelated) contribution of $X_{i}$. 

\begin{figure}
\centering
\begin{subfigure}{0.48\textwidth}
\centering
\includegraphics[width = \textwidth]{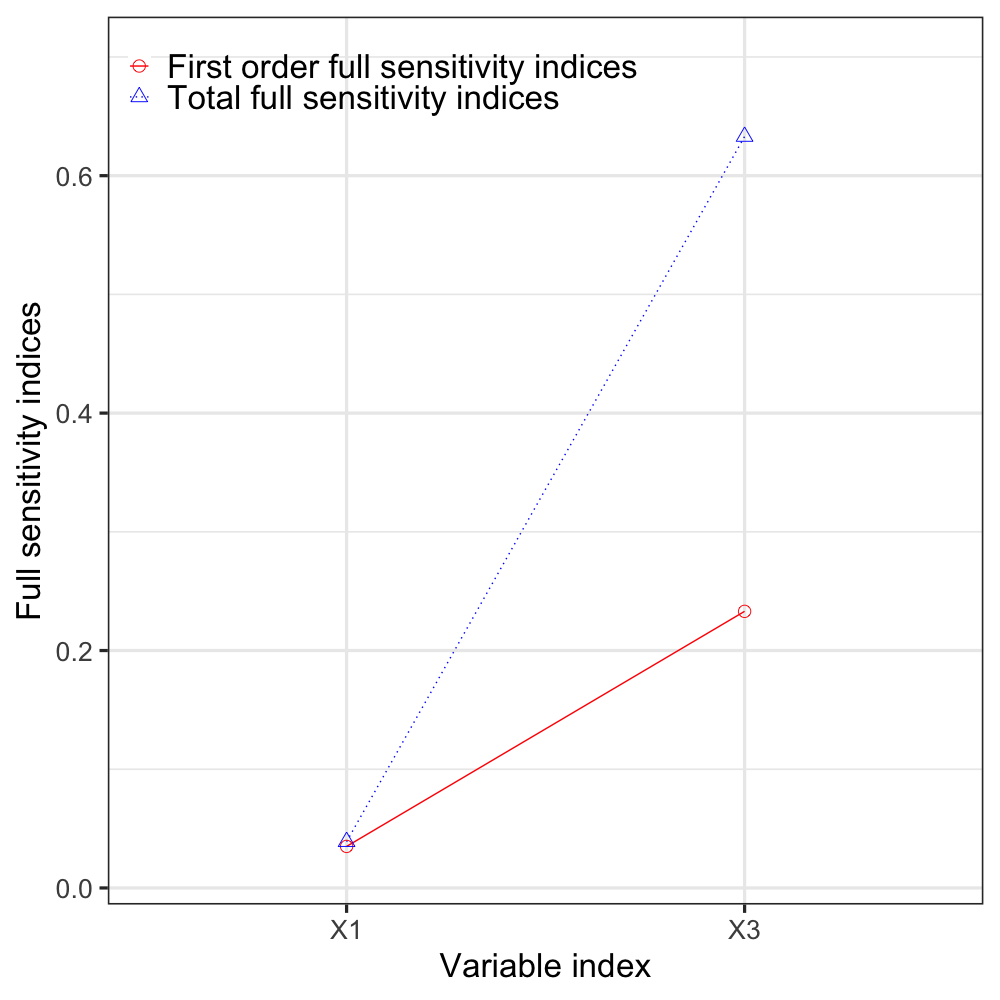}
\label{fig:left}
\end{subfigure}
\begin{subfigure}{0.48\textwidth}
\centering
\includegraphics[width = \textwidth]{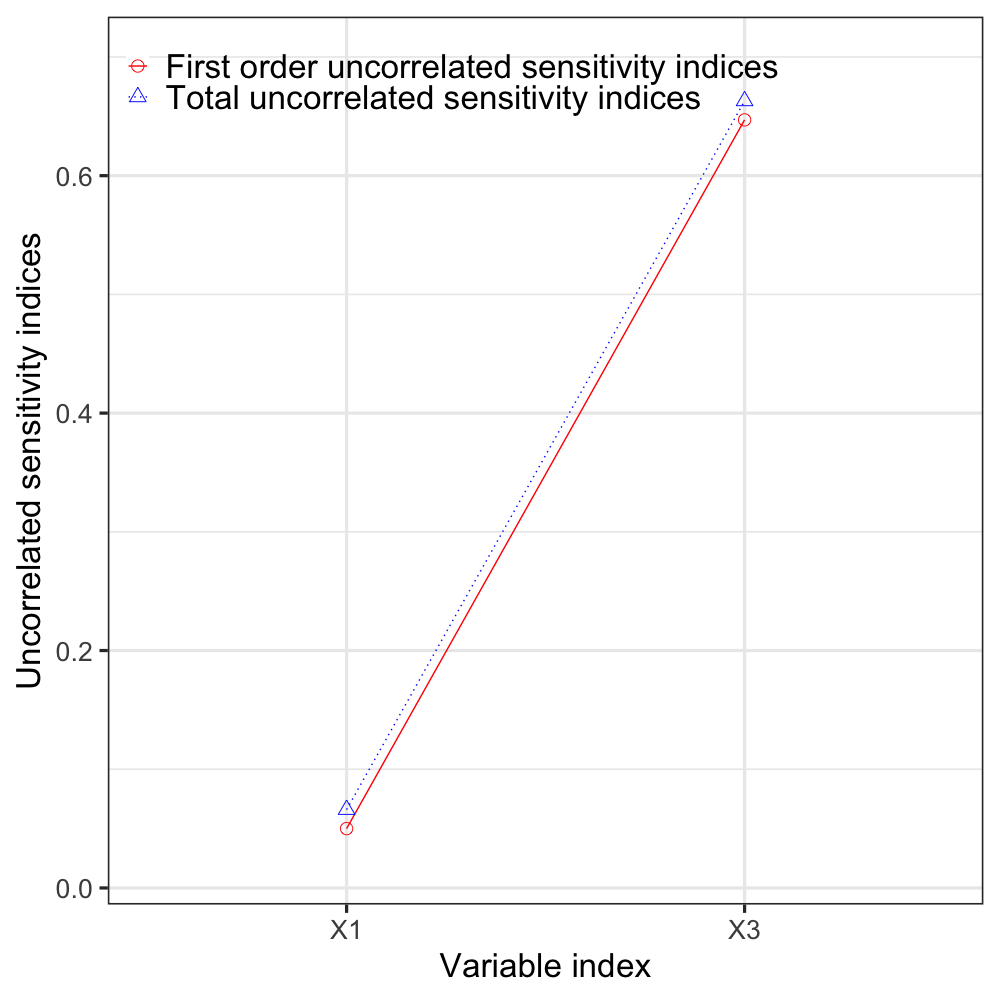}
\label{fig:right}
\end{subfigure}
\caption{The left-hand side graph shows the \textit{full} sensitivity indices for dependent input variables and the right-hand side graph shows the \textit{uncorrelated} sensitivity indices for dependent input variables in Example 2.  }
\label{fig:exampel2}
\end{figure}

\subsection{Example 3}
For the third example, we modified the example in \cite{jacques:2006} to have a more complex structure and involve multiple types of probability distributions as follows: 
\begin{equation}
\begin{aligned}
\label{eq:sim2formula}
&\begin{pmatrix}
X_{1}\\
X_{2}\\
X_{3}\\
X_{4}
\end{pmatrix} & \sim & \mathcal{N} \left[\left(\begin{array}{c}
0\\
0\\
0 \\
0
\end{array}\right),\left(\begin{array}{cccc}
1 & 0 & 0 & 0\\
0 & 1 & 0 & 0\\
0 & 0 & 1 & 0.3\\
0 & 0 & 0.3 & 1\\
\end{array}\right)\right], \\
& X &\sim &\mathcal{U}(0,\,1), \\ 
& X_{5} &= &\theta_{1}X + \mathcal{U}(0,\,1), \\
& X_{6} &= &\theta_{2}X + \theta_{3}X^{2} + \mathcal{U}(0,\,1), \\
&Y &=& X_{1}X_{2} + X_{3}X_{4} + X_{5}X_{6}.
\end{aligned}
\end{equation}
Here, $(X_{1}, X_{2}, X_{3}, X_{4})$ follows a multivariate Gaussian distribution with the parameters as above. The inputs $X_{5}$ and $X_{6}$ are dependent on each other, but their dependency cannot be explained by their first-order conditional moments. In this experiment, we set $(\theta_1,\theta_2,\theta_3) = (0.4, 0.6, 1)$ and obtain 10,000 random observations. 
Because the cumulative sum of the first two conditional order-based sensitivity indices is $\sum_{i=1}^{2}\tilde{S}_{i}=1$, we exclude third- and higher-order interaction terms in the PCE model to make it sparse. As for the two-way interaction terms, as shown in Figure \ref{fig:two-way}, most of the corresponding PCE coefficients are nearly zero except for the polynomial terms, $X_{1}X_{2}$, $X_{3}X_{4}$, and $X_{5}X_{6}$. It indicates that $X_{1}X_{2}$, $X_{3}X_{4}$, and $X_{5}X_{6}$ are the only interaction terms in the true model. Various sparse PCE approaches \cite{blatman2008sparse,blatman2011adaptive, jakeman2015enhancing, peng2016polynomial} can be additionally applied here to construct a sparser PCE model with only the important orthonormal polynomials of inputs.


\begin{figure}
\centering
\vspace*{-2.5cm}
\includegraphics[width = .8\textwidth]{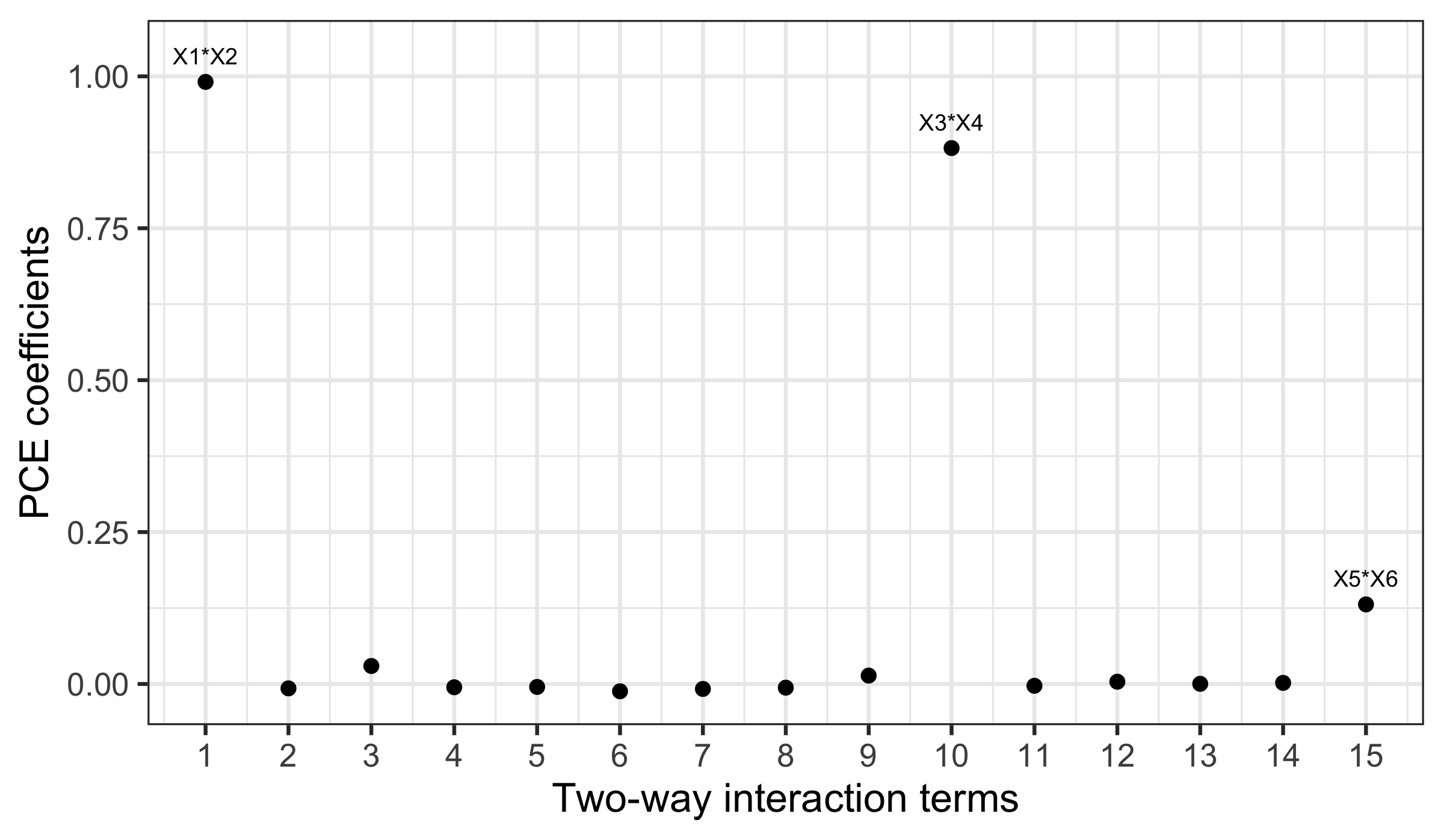}
\label{fig:right2}
\caption{PCE coefficients v.s. the two-way interaction terms in the PCE model in Example 3. The significant interaction terms, $X_{1}X_{2}$, $X_{3}X_{4}$, and $X_{5}X_{6}$, are identified.}
\label{fig:two-way}
\end{figure}


Because $\{X_{1}, X_{2}\}$, $\{X_{3}, X_{4}\}$, and $\{X_{5}, X_{6}\}$ are mutually independent, we can infer from the conditional order-based sensitivity indices that the output $Y$ is composed of three non-interacting functions $f_{12}(X_{1}, X_{2})$, $ f_{34}(X_{3}, X_{4})$, and $f_{56}(X_{5}, X_{6})$. Thanks to this special structure, we can directly calculate the total sensitivity indices for $\{X_{1}, X_{2}\}$, $\{X_{3}, X_{4}\}$, and $\{X_{5}, X_{6}\}$. Without permuting the order of the input variables, the total Sobol indices can be calculated as follows: 
\begin{equation}
\begin{aligned}
\label{eq:exeq1}
ST_{\left\{X_{1}, X_{2}\right\}} = \overbar{ST}_{X_{1}} + \overbar{ST}_{X_{2}}, \\
ST_{\left\{X_{3}, X_{4}\right\}} = \overbar{ST}_{X_{3}} + \overbar{ST}_{X_{4}}, \\
ST_{\left\{X_{5}, X_{6}\right\}} = \overbar{ST}_{X_{5}} + \overbar{ST}_{X_{6}}, \\
\end{aligned}
\end{equation}
where $\overbar{ST}_{X_{i}}$ is the conditional \textit{total full} sensitivity index defined in Section~\ref{sec:full_SI}.  

We validate the sensitivity indices from the proposed method with the values from an analytical method (see Appendix {\color{black}A.2}). We also calculate $ST_{\{X_{1}, X_{2}\}}$ and $ST_{\{X_{3}, X_{4}\}}$ using the benchmark method in \cite{mara:2012} assuming the knowledge that $\left(X_{1}, X_{2}, X_{3}, X_{4}\right)$ is multivariate Gaussian distributed and $\{X_{1}, X_{2}\}$, $\{X_{3}, X_{4}\}$, and $\{X_{5}, X_{6}\}$ are mutually independent. Then, $ST_{\{X_{5}, X_{6}\}}$ is calculated based on the fact that $ST_{\{X_{1}, X_{2}\}} + ST_{\{X_{3}, X_{4}\}} +ST_{\{X_{5}, X_{6}\}} = 1$.

\begin{table}[ht]
\centering
\caption{Sample means of sensitivity indices and $95\%$ confidence intervals.} 
\label{tab:sens-sim2}
\small
\begin{tabular}{|  c| c c c |  }
 \hline
\backslashbox{Method}{Input set} & $ST_{\{X_{1}, X_{2}\}}$  &$ST_{\{X_{3}, X_{4}\}}$& $ST_{\{X_{5}, X_{6}\}}$  \\
 \hline
Analytical method & {\color{black}0.402}&{\color{black}0.438}& {\color{black}0.160}\\
Proposed method$^*$ & {\color{black}0.402} &{\color{black} 0.438} & {\color{black}0.160}\\ 
& {\color{black}(0.401, 0.404)} &{\color{black} (0.437, 0.439)} & {\color{black}(0.159, 0.160)}\\
Benchmark method\cite{mara:2012} &{\color{black}$0.403^{\dagger}$}&{\color{black}$0.439^{\dagger}$}&{\color{black}$(0.158)^{\ddagger}$}\\
& {\color{black}(0.402, 0.404)}&{\color{black}(0.438, 0.440)}& {\color{black}(0.157, 0.160)}\\
\hline
\end{tabular}
\begin{tablenotes}
      \small
      \item \textit{Note:} {$\dagger$The value is obtained using the sample variance to estimate $Var(Y)$ in the denominator of the sensitivity index instead of using the PCE coefficients from the benchmark method (see Eq.~\eqref{eq:mom}) because the latter estimation suffers a non-negligible bias in this example that does not satisfy the assumption of the benchmark method. $\ddagger$The value cannot be obtained directly from the benchmark method, but we calculate the value based on the assumption that the user knows that $X_{5}$ and $X_{6}$ are independent of the rest of the inputs and that $\left(X_{1}, X_{2}, X_{3}, X_{4}\right)$ follows a multivariate Gaussian distribution. *The proposed method does not require any assumption. The sample means of sensitivity indices and $95\%$ bootstrap confidence intervals 
(using 10,000 bootstrap samples) are calculated based on 500 replications. In each replication, sensitivity indices are calculated using 5,000 random observations of the inputs and output in {\color{black}Eq.~}(\ref{eq:sim2formula}).} 
    \end{tablenotes}
\end{table}

As shown in Table \ref{tab:sens-sim2}, the sensitivity indices from our method are close to . those from the benchmark method and the analytical values. Note that, in contrast to the benchmark method, the proposed method is a data-driven approach that does not impose any assumption on the inputs. 

\subsection{Example 4}
{
\color{black}
As the fourth example, we consider the 23-bar horizontal truss example in \cite{torre2019data}. The output of interest, $Y$, is a downward vertical displacement at the mid span of the structure subject to random loads. As depicted in Figure \ref{fig:truss}, the uncertainty of $Y$ depends on the ten random inputs in $\boldsymbol{X}=\left(E_1, E_2, A_1, A_2, P_1, \ldots, P_6\right)$: namely, uncertain Young modulus $E_{i}, i=1,2,$ and uncertain cross-sectional area $A_{i}, i=1,2,$ for two different groups of bars (horizontal for $i=1$ and diagonal for $i=2$), and the random loads $P_{i}, i=1,2,\cdots,6$.  The inputs $E_{i}, i=1,2,$ and $A_{i}, i=1,2,$ are assumed to be mutually independent and follow the following distributions:
\begin{equation}
    \begin{aligned}
    E_{1}, E_{2} \sim \mathcal{LN}(2.1 \times 10^{11}, 2.1\times 10^{10}) \; \textrm{[Pa]},\\
    A_{1} \sim \mathcal{LN}(2.0 \times 10^{-3}, 2.0 \times 10^{-4}) \; [\textrm{m}^2],\\
    A_{2} \sim \mathcal{LN}(1.0 \times 10^{-3}, 1.0 \times 10^{-4}) \;  [\textrm{m}^2],
    \end{aligned} \nonumber
\end{equation}
where $\mathcal{LN}(\mu,\sigma)$ denotes the lognormal distribution with mean $\mu$ and standard deviation $\sigma$.

The dependent inputs $P_{i}, i=1,2,\cdots,6,$ have the following Gumbel marginal distribution function with mean $\mu = 5\times 10^{4} \;\textrm{[N]}$ and standard deviation $\sigma = 7.5\times 10^{3} \;\textrm{[N]}$:
\begin{equation}
\begin{aligned}
 F_{i}(x;\alpha, \beta) = e^{-e^{-(x-\alpha)}/\beta}, i =1,2,\ldots,6,
\end{aligned}\nonumber
\end{equation}
where $\beta = \sqrt{6}\sigma/\pi$, $\alpha = \mu - \gamma \beta$, and $\gamma \approx 0.5772$ is the Euler-Masch{\color{black}e}roni constant. The dependence between the inputs $P_{i}, i=1,\ldots,6$, is encoded in the C-vine copula with the following density:
\begin{equation}
\label{eq:coupula}
c_{\boldsymbol{X}}^{(\mathcal{G})}(u_{1},\ldots,u_{6}) = \prod_{j=2}^{6}c_{1j;\theta=1.1}^{(\mathcal{G}\mathcal{H})}(u_{1},u_{j}),    
\end{equation}
where $c_{1j;\theta=1.1}^{(\mathcal{G}\mathcal{H})}$ is the density of the pair-copula between $P_{1}$ and $P_{j}$, $j=2,\ldots,6$. $\mathcal{G}\mathcal{H}$ represents the Gumbel-Hougaard family whose bivariate copula is  $$C^{(\mathcal{G}\mathcal{H})}_\theta(u,v) =\exp\left(-\left( \left(-\log u\right)^\theta + \left(- \log v\right)^\theta \right)^{1/\theta}\right), \quad \theta \in \left[1, \infty\right).$$ The parameter $\boldsymbol{\theta}$ decides the dependence between two loads (i.e., the larger parameter $\boldsymbol{\theta}$ the stronger dependence). 
In this case, $P_{1}$ is equally and positively correlated with each of $P_{2},\ldots,P_{6}$. Thus, $P_{2},\ldots,P_{6}$ are positively correlated with each other although they are conditionally independent given $P_{1}$ (see a sample correlation matrix for the loads in Appendix A.3). The output $Y$ is simulated using the response surface model (see Appendix A.4) in \cite{lee2006response}, which was constructed based on a finite element analysis. The explicit relationship between $Y$ and $\boldsymbol{X}$ in the response surface model allows us to evaluate the estimated sensitivity indices.

\begin{figure}
\centering
\vspace*{-2.5cm}
\includegraphics[width = .85\textwidth]{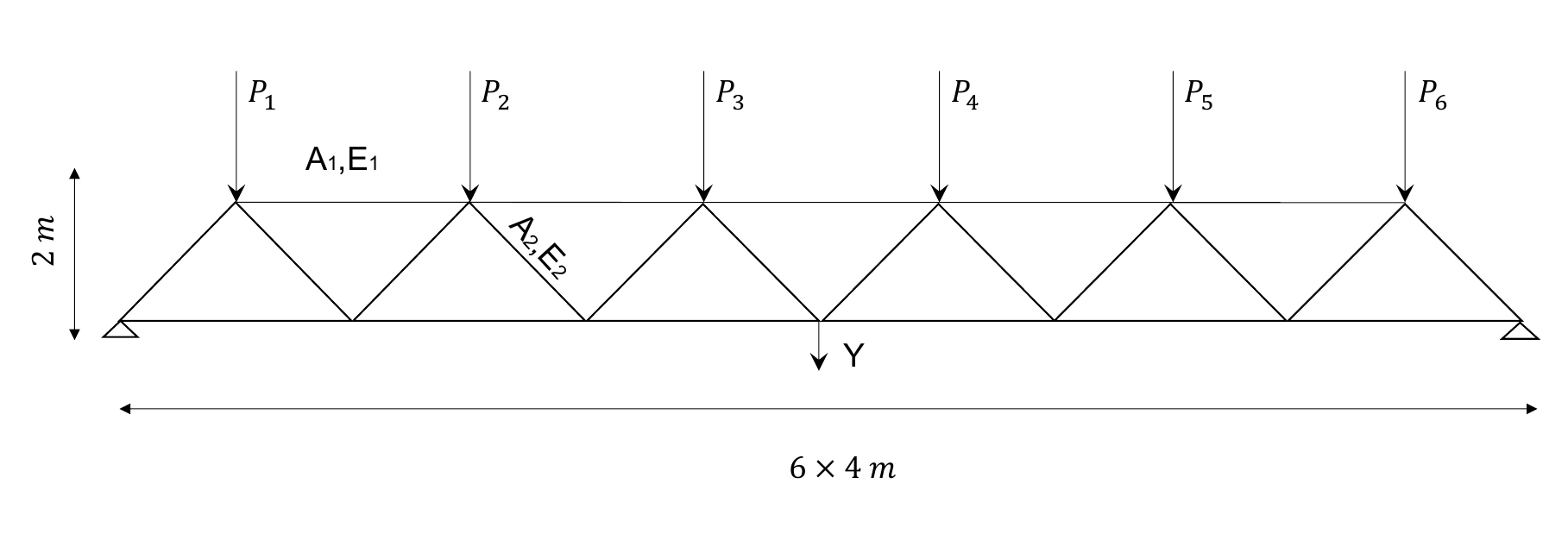}
\caption{{\color{black}Scheme of the horizontal truss model modified from \cite{lee2006response}. The downward vertical displacement at the mid span of the structure $Y$ depends on Young modulus $E_{i}, i=1,2$, cross-sectional area $A_{i}, i=1,2$ for both horizontal and diagonal bars, and the random loads $P_{i}, i=1,2,\cdots,6$.}}
\label{fig:truss}
\end{figure}

This realistic problem {\color{black}with dependent inputs} has neither analytically known sensitivity indices nor any benchmark methods that attempted to estimate the sensitivity indices {\color{black}(see \cite{blatman2011adaptive} for a related sensitivity analysis with \textit{independent} inputs)}. Implementing a brute-force Monte Carlo approach is computationally challenging, if not infeasible, because the analytical expressions of the sensitivity indices involve variances of conditional expectations that condition on (multiple combinations of) multiple inputs. A similar challenge lies in even estimating Sobol indices for \emph{independent} inputs and is studied extensively in the literature \cite{myshetskaya2008monte,tissot2015randomized,owen2013better}. No Monte Carlo method has satisfactorily addressed the computational challenge yet. Extending the existing Monte Carlo methods for independent inputs to handle dependent inputs is left for future work. 

In this study, we examine the estimated sensitivity indices to confirm that they agree with their expected physical interpretations.
As it is shown in Table \ref{tab:sens-sim4}, $P_2$ (resp. $P_{3}$) and $P_{5}$ (resp. $P_{4}$) have almost the same sensitivity indices for $\bar{S}$, $\overbar{ST}$, $S^{u}$, and $ST^{u}$. This can be explained by a) the physical symmetry between $P_2$ (resp. $P_{3}$) and $P_{5}$ (resp. $P_{4}$) with respect to the location at which the output $Y$ is measured (see Figure \ref{fig:truss} and the response surface model in Appendix A.4) and b) their symmetric correlations with other inputs (recall the copula density in Eq.~\eqref{eq:coupula} and see Appendix A.3). 
On the other hand, the differences between the \emph{full} sensitivity indices (i.e., $\bar{S}$ and $\overbar{ST}$) for $P_{1}$ and $P_{6}$ can be explained by the fact that $P_{1}$ has over 5 times stronger correlations than $P_{6}$ with $P_{2},\ldots, P_{5}$ (see Appendix A.3). In contrast, the \emph{uncorrelated} sensitivity indices (i.e., $S^{u}$ and $ST^{u}$) for $P_{1}$ and $P_{6}$ are the same (up to 4 decimal places) because the uncorrelated effects of $P_{1}$ and $P_{6}$ on $Y$ should be very similar (see the coefficients of the response surface model in Appendix A.4.). The sensitivity indices for all the other inputs are similarly confirmed to be consistent with their expected physical interpretations based on the response surface model, which reflects the physical relationship between $Y$ and $\boldsymbol{X}$, and the dependence structure of $\boldsymbol{X}$. 

{\color{black}In addition, although not directly comparable due to different settings, still the first and total uncorrelated sensitivity indices (i.e., $S^{u}$ and $ST^{u}$) have similar magnitudes as the first and total Sobol indices (i.e., $S$ and $ST$) reported in Table 5 of \cite{le2017metamodel} and Table 2 of \cite{blatman2011adaptive}, respectively, for all ten inputs. Both articles \cite{le2017metamodel,blatman2011adaptive} assumed that $P_1$ through $P_6$ are \textit{independent}, and directly computed $Y$ using a finite element model.}



\begin{table}[]
\centering
\caption{{\color{black}Sample means of sensitivity indices and $95\%$ confidence intervals.}}
\label{tab:sens-sim4}
\begin{adjustbox}{width=0.98\textwidth}
\small
\begin{tabular}{|>{\color{black}}c|>{\color{black}}c>{\color{black}}c>{\color{black}}c>{\color{black}}c|}
\hline
\backslashbox{Input}{Sensitivity indices}       & $\bar{S}$                     & $\overbar{ST}$                   & $S^{u}$              & $ST^{u}$             \\ \hline
$E_{1}$ & 0.324  (0.321, 0.327) & 0.371 (0.367, 0.374) & 0.286 (0.284, 0.288) & 0.312 (0.310, 0.314) \\  \hline
$E_{2}$ & 0.013 (0.012, 0.014)  & 0.036 (0.035, 0.037) & 0.009 (0.008, 0.009) & 0.009 (0.008, 0.009) \\  \hline
$A_{1}$ & 0.325 (0.322, 0.328)  & 0.370 (0.367, 0.374) & 0.285 (0.283, 0.287) & 0.310 (0.308, 0.312) \\  \hline
$A_{2}$ & 0.013 (0.012, 0.014)  & 0.037 (0.036, 0.038) & 0.008 (0.008, 0.008) & 0.008 (0.008, 0.008) \\  \hline
$P_{1}$ & 0.065 (0.063, 0.068)  & 0.096 (0.093, 0.099) & 0.004 (0.004, 0.004) & 0.004 (0.004, 0.004) \\   \hline
$P_{2}$ & 0.060 (0.058, 0.062)  & 0.088 (0.085, 0.090) & 0.033 (0.033, 0.033) & 0.035 (0.035, 0.036) \\   \hline
$P_{3}$ & 0.105 (0.103, 0.108)  & 0.135 (0.132, 0.138) & 0.068 (0.067, 0.068) & 0.073 (0.072, 0.073) \\    \hline
$P_{4}$ & 0.102 (0.010, 0.105)  & 0.130 (0.128, 0.133) & 0.068 (0.067, 0.068) & 0.073 (0.072, 0.073) \\   \hline
$P_{5}$ & 0.057 (0.055, 0.059)  & 0.084 (0.082, 0.087) & 0.033 (0.033, 0.033) & 0.035 (0.035, 0.036) \\   \hline
$P_{6}$ & 0.017 (0.016, 0.018)  & 0.043 (0.042, 0.045) & 0.004 (0.004, 0.004) & 0.004 (0.004, 0.004) \\  \hline
\end{tabular}
\end{adjustbox}
\begin{tablenotes}
      \small
      \item {\color{black}\textit{Note:} {The sample means of sensitivity indices and $95\%$ bootstrap confidence intervals (using $10,000$ bootstrap samples) are calculated based on $500$ replications. In each replication, sensitivity indices are calculated using $500$ random observations of the inputs and outputs in Eq.~\eqref{eq:truss_out}.} } 
    \end{tablenotes}
\end{table}

}

\section{Conclusion}
\label{sec:5}
In this paper, data-driven sensitivity indices for a model with \textit{dependent} inputs are proposed using the PCE without imposing any strong assumptions on the model inputs. The modified Gram-Schmidt algorithm with the empirical measure is utilized to construct orthonormal polynomials for a PCE model on the merit of numerical stability. The proposed data-driven method {\color{black}yields} the \textit{full} sensitivity indices and the \textit{uncorrelated} sensitivity indices by constructing ordered partitions of orthonormal polynomials of inputs for a PCE model. The proposed conditional order-based sensitivity indices for a model with \textit{dependent} inputs help reduce the complexity of a PCE model while keeping the effect hierarchy principle. {\color{black}Four} numerical examples  validate the proposed method. 

The proposed method requires polynomials of inputs, which are fed into the modified Gram-Schmidt algorithm, to be linearly independent. This suggests a future research direction because there are multiple ways of constructing linearly independent polynomials from a linearly dependent polynomial basis. How to build a theoretically and practically desirable basis warrants more investigation.   

\section*{Appendix} 

\subsection*{A.1. List of sensitivity indices}
\glsaddall
\newglossaryentry{aaaa}{
  name = {$S_{X_{i}}$} , 
  description = {First-order Sobol index of the output $Y$ with respect to the input $X_{i}$}
}
\newglossaryentry{bbbb}{
  name = {$ST_{X_{i}}$},
  description = {Total Sobol index of the output $Y$ with respect to the input $X_{i}$}
}

\newglossaryentry{cccc}{
  name = {$\bar{S}_{X_{i}}$}, 
  description = {First-order \textit{full} sensitivity index of the output $Y$ with respect to the input $X_{i}$}
}

\newglossaryentry{dddd}{
  name = {$\overbar{ST}_{X_{i}}$}, 
  description = {Total \textit{full} sensitivity index of the output $Y$ with respect to the input $X_{i}$}
}

\newglossaryentry{eeee}{
  name = {$S^{u}_{X_{i}}$}, 
  description = {First-order \textit{uncorrelated} sensitivity index of the output $Y$ with respect to the input $X_{i}$}
}
\newglossaryentry{ffff}{
  name = {$ST^{u}_{X_{i}}$},
  description = {Total \textit{uncorrelated} sensitivity index of the output $Y$ with respect to the input $X_{i}$}
}

\newglossaryentry{gggg}{
  name = {$\tilde{S}_{k\vert 1,2,\ldots,k-1}$},
  description = {The $k^{th}$ order sensitivity index of the output $Y$ with respect to the inputs $\boldsymbol{X}$ after taking account of the first order sensitivity index through the $(k-1)^{th}$ order sensitivity index}
}
\printglossaries
\subsection*{A.2. Analytical method for calculating the sensitivity indices in Example 3 in Section~\ref{sec:4}} 
The following lemma is used for the analytical method in Example 3 in Section~\ref{sec:4}. 
\begin{lem}
\label{lem1}
Suppose 
\begin{equation*}
\begin{aligned}
&\begin{pmatrix}
X_{1}\\
X_{2}\\
\end{pmatrix} & \sim & \mathcal{N} \left[\left(\begin{array}{c}
\mu_{1}\\
\mu_{2}\\
\end{array}\right),\left(\begin{array}{cc}
\sigma_{1}^{2} & \rho\sigma_{1}\sigma_{2}  \\
\rho\sigma_{1}\sigma_{2} & \sigma_{2}^{2} \\
\end{array}\right)\right]. \\
\end{aligned}
\end{equation*}
\normalfont Then, $Var(X_{1}X_{2}) = \mu_{1}^{2}\sigma_{2}^{2}  + \mu_{2}^{2}\sigma_{1}^{2} + \sigma_{1}^{2}\sigma_{2}^{2} + 2 \mu_{1}\mu_{2}\rho\sigma_{1}\sigma_{2} + \rho^{2}\sigma_{1}^{2}\sigma_{2}^{2}$. 
\end{lem}

\begin{proof} 
We can express $X_{1}$ and $X_{2}$ as follows: %
\begin{equation*}
\begin{aligned}
X_{1} = \mu_{1} + r\sigma_{1}Z + (1-r^{2})^{\frac{1}{2}}\sigma_{1}Y_{1}, \\ 
X_{2} = \mu_{2} + r\sigma_{2}Z + (1-r^{2})^{\frac{1}{2}}\sigma_{2}Y_{2}, \\ 
\end{aligned}
\end{equation*}
where $r = \sqrt{\rho}$, $Z \sim \mathcal{N}(0,1)$, and $Y_{i} \sim \mathcal{N}(0,1), i = 1,2$. We have the covariance among $Z$ and $Y_{i}, i = 1,2$ as follows: 
\begin{equation*}
\begin{aligned}
Cov(Z, Y_{i}) &=0, \text{ for } i = 1,2 ,\\
Cov(Y_{1}, Y_{2}) &= 0. 
\end{aligned}\nonumber
\end{equation*}
Therefore, we have
\begin{equation*}
\begin{aligned}
\label{eq:prooflem3}
Var(X_{1}X_{2}) &= E(X_{1}^{2}X_{2}^{2}) - [E(X_{1}X_{2})]^{2} \\
&= E(X_{1}^{2})E(X_{2}^{2}) + Cov(X_{1}^{2}, X_{2}^{2}) - [E(X_{1})E(X_{2}) + Cov(X_{1}, X_{2})]^{2} \\ 
&= (\sigma_{1}^{2} + \mu_{1}^{2})(\sigma_{2}^{2} + \mu_{2}^{2}) + Cov(r^{2}\sigma_{1}^{2}Z^{2} + 2\mu_{1}r\sigma_{1}Z, r^{2}\sigma_{2}^{2}Z^{2} + 2\mu_{2}r\sigma_{2}Z)- \\
& \quad  (\mu_{1}\mu_{2} + r^{2}\sigma_{1}\sigma_{2})^{2} \\
&= \mu_{1}^{2}\sigma_{2}^{2} + \mu_{2}^{2}\sigma_{1}^{2} + \sigma_{1}^{2}\sigma_{2}^{2} + 2\mu_{1}\mu_{2}\rho\sigma_{1}\sigma_{2} + \rho^{2}\sigma_{1}^{2}\sigma_{2}^{2}.
\end{aligned}
\end{equation*}
\end{proof}
\noindent

We now present how to analytically calculate the sensitivity indices in Example 3. 
From (\ref{eq:sim2formula}), $\{X_{1}, X_{2}\}$, $\{X_{3}, X_{4}\}$, and $\{X_{5}, X_{6}\}$ are mutually independent. We have 
\begin{equation*}
\begin{aligned}
\label{eq:appen2}
Var(Y) &= Var(X_{1}X_{2} + X_{3}X_{4} + X_{5}X_{6}) \\
&= Var(X_{1}X_{2}) + Var(X_{3}X_{4}) + Var(X_{5}X_{6}).
\end{aligned}
\end{equation*}

Based on Lemma \ref{lem1}, we can easily obtain $Var(X_{1}X_{2})$ and $Var(X_{3}X_{4})$. As for $Var(X_{5}X_{6})$, $X_{5}$ and $X_{6}$ can be expressed as follows: 
\begin{equation*}
\begin{aligned}
\label{eq:appendfin1}
X_{5} &= \theta_{1}U_{1} + U_{2}, \\
X_{6} &= \theta_{2}U_{1} + \theta_{3}U_{1}^{2} + U_{3}, 
\end{aligned}
\end{equation*}
where $U_{i} \sim U(0, 1), i=1,2,3$ and $U_{i}$'s are mutually independent for $i = 1,2,3$. Because
\begin{equation*}
\begin{aligned}
\label{eq:appendfin2}
Var(X_{5}X_{6}) &= E(X_{5}^{2}X_{6}^{2}) - [E(X_{5}X_{6})]^{2} \\
&= Cov(X_{5}^{2},X_{6}^{2}) + E(X_{5}^{2})E(X_{6}^{2}) - [Cov(X_{5},X_{6}) + E(X_{5})E(X_{6})]^{2} ,
\end{aligned}
\end{equation*}
using the property that  $U_{i}$'s are mutually independent for $i = 1,2,3$, it is straightforward to express $Var(X_{5}X_{6})$ as a function of $(\theta_1,\theta_2,\theta_3)$ and the moments of $U_{i}, i=1,2,3$. 

After calculating $Var(X_{1}X_{2})$, $Var(X_{3}X_{4})$, and $Var(X_{5}X_{6})$, we can calculate $ST_{\{X_{1}X_{2}\}}$, $ST_{\{X_{3}X_{4}\}}$, and $ST_{\{X_{5}X_{6}\}}$ as follows: 

\begin{equation*}
\begin{aligned}
\label{eq:appendfin3}
ST_{\{X_{1}X_{2}\}} &= \frac{Var(X_{1}X_{2})}{Var(X_{1}X_{2})+Var(X_{3}X_{4})+Var(X_{5}X_{6})}, \\
ST_{\{X_{3}X_{4}\}} &= \frac{Var(X_{3}X_{4})}{Var(X_{1}X_{2})+Var(X_{3}X_{4})+Var(X_{5}X_{6})}, \\
ST_{\{X_{5}X_{6}\}} &= \frac{Var(X_{5}X_{6})}{Var(X_{1}X_{2})+Var(X_{3}X_{4})+Var(X_{5}X_{6})}. \\
\end{aligned}
\end{equation*}
{
\color{black}
\subsection*{A.3. The correlation matrix for the loads in Example 4 in Section~\ref{sec:4}}
The correlation matrix for the loads listed below is estimated using $10^{6}$ random observations.

\begin{table}[h]
\centering
\caption{{\color{black}A correlation matrix for the six loads estimated based on $10^{6}$ random observations generated based on the C-vine copula in Eq.~\eqref{eq:coupula}.}}
\begin{tabular}{|>{\color{black}}c|>{\color{black}}c>{\color{black}}c>{\color{black}}c>{\color{black}}c>{\color{black}}c>{\color{black}}c|}
\hline
        & $P_{1}$ & $P_{2}$ & $P_{3}$ & $P_{4}$ & $P_{5}$ & $P_{6}$ \\ \hline
$P_{1}$ & 1.000   & 0.172   & 0.173   & 0.171   & 0.176   & 0.173   \\
$P_{2}$ & 0.172   & 1.000   & 0.032   & 0.031   & 0.033   & 0.032   \\
$P_{3}$ & 0.173   & 0.032   & 1.000   & 0.032   & 0.034   & 0.031   \\
$P_{4}$ & 0.171   & 0.031   & 0.032   & 1.000   & 0.033   & 0.031   \\
$P_{5}$ & 0.176   & 0.033   & 0.034   & 0.033   & 1.000   & 0.032   \\
$P_{6}$ & 0.173   & 0.032   & 0.031   & 0.031   & 0.032   & 1.000   \\ \hline
\end{tabular}
\end{table}

\subsection*{A.4. The response surface model used for simulating the output $Y$ in Example 4 in Section~\ref{sec:4}}
The response surface model used to simulate the output $Y$ is provided in \cite{lee2006response} as follows:
\begin{equation}
\label{eq:truss_out}
\begin{aligned}
Y &= 2.8070 + 1.2598E^{\prime}_{1} + 0.2147E^{\prime}_{2} + 1.2559A^{\prime}_{1} + 0.2133A^{\prime}_{2} - 0.1510P^{\prime}_{1} - 0.4238P^{\prime}_{2} - \\ &0.6100P^{\prime}_{3} - 0.6100P^{\prime}_{4} -0.4238P^{\prime}_{5} - 0.1510P^{\prime}_{6} - 0.1978E^{\prime2}_{1} - 0.0362E^{\prime2}_{2} - 0.2016A^{\prime2}_{1} - \\ & 0.0346A^{\prime2}_{2}+0.0023P^{\prime2}_{1}+0.0008P^{\prime2}_{2}+0.0036P^{\prime2}_{3}+ 0.0036P^{\prime2}_{4}+0.0008P^{\prime2}_{5}+0.0023P^{\prime2}_{6}- \\
& 0.0042E^{\prime}_{1}E^{\prime}_{2} -0.3022E^{\prime}_{1}A^{\prime}_{1}-0.0110E^{\prime}_{1}A^{\prime}_{2}+0.0381E^{\prime}_{1}P^{\prime}_{1}+0.0871E^{\prime}_{1}P^{\prime}_{2}+0.1232E^{\prime}_{1}P^{\prime}_{3}+ \\
& 0.1232E^{\prime}_{1}P^{\prime}_{4}+0.0871E^{\prime}_{1}P^{\prime}_{5}+0.0346E^{\prime}_{1}P^{\prime}_{6}+0.0041E^{\prime}_{2}A^{\prime}_{1}+0.0110A^{\prime}_{1}A^{\prime}_{2}+0.0261A^{\prime}_{1}P^{\prime}_{1}+ \\
& 0.0831A^{\prime}_{1}P^{\prime}_{2}+0.1172A^{\prime}_{1}P^{\prime}_{3} + 0.1172A^{\prime}_{1}P^{\prime}_{4} + 0.0832A^{\prime}_{1}P^{\prime}_{5} + 0.0296A^{\prime}_{1}P^{\prime}_{6},
    \end{aligned}
\end{equation}
where $E^{\prime}_{i}, i =1,2$, $A^{\prime}_{i}, i=1,2$, and $P^{\prime}_{i}, i=1,2,3,4,5,6$ are the standardized inputs. For example, $E^{\prime}_{1} = \frac{E - \mu_{E_{1}}}{\sigma_{E_{1}}}$, where $\mu_{E_{1}}$ is the mean of $E_{1}$ and $\sigma_{E_{1}}$ is the standard deviation of $E_{1}$.
}
\section*{Acknowledgements}
{\color{black}The authors would like to thank the Editor and two anonymous reviewers for their feedback that helped significantly improve this article.
This work was supported in part by the National Science Foundation (NSF) under Grant CMMI-1824681}.

\section*{References}

\bibliography{mybibfile}

\end{document}